\RequirePackage{latexml}
\pdfoutput=1
\iflatexml

	\documentclass{article}
	\author{%
		Dominik Peters \\
		University of Oxford
	}
	\usepackage{amsthm,amssymb,natbib}
	\newtheorem{theorem}{Theorem}[section]
	\newtheorem{lemma}[theorem]{Lemma}
	
	\newtheorem{proposition}[theorem]{Proposition}

\else

	\documentclass[sigconf]{aamas}  
	\setcopyright{ifaamas}  
	\acmDOI{}  
	\acmISBN{}  
	\acmConference[]{}  
	\acmYear{}  
	\copyrightyear{}  
	\acmPrice{}  
	
	\author{Dominik Peters}
	\affiliation{%
		\institution{University of Oxford}
		\city{Oxford} 
		\country{UK}}
	\email{dominik.peters@cs.ox.ac.uk}

	
	\thanks{Version v3, November 2024. \\ Original version published at AAMAS 2018. \\ Current affiliation: CNRS, LAMSADE, Université Paris Dauphine - PSL. \\ Email: dominik.peters@lamsade.dauphine.fr}

\fi

\usepackage{booktabs,amsmath,mathtools,xcolor,microtype}

\usepackage{tikz}
\usetikzlibrary{positioning,quotes}

\renewcommand*{\le}{\leqslant}
\renewcommand*{\ge}{\geqslant}

\newcommand{\bigland}{\bigwedge} 
\newcommand{\biglor}{\bigvee}

\usepackage{enumitem}
\setlist[description]{leftmargin=2\parindent,labelindent=\parindent}

\usepackage[ruled,noend]{algorithm2e} 
\SetKwInput{KwOut}{Question}

\SetAlFnt{\small}
\SetAlCapFnt{\small}
\SetAlCapNameFnt{\small}
\SetAlCapHSkip{0pt}
\IncMargin{-\parindent}

\begin{document}

\title{Proportionality and Strategyproofness in Multiwinner Elections}

\begin{abstract}  
Multiwinner voting rules can be used to select a fixed-size committee from a larger set of candidates. We consider approval-based committee rules, which allow voters to approve or disapprove candidates. In this setting, several voting rules such as Proportional Approval Voting (PAV) and Phragm\'en's rules have been shown to produce committees that are proportional, in the sense that they proportionally represent voters' preferences; all of these rules are strategically manipulable by voters. On the other hand, a generalisation of Approval Voting gives a non-proportional but strategyproof voting rule. We show that there is a fundamental tradeoff between these two properties: we prove that no multiwinner voting rule can simultaneously satisfy a weak form of proportionality (a weakening of justified representation) and a weak form of strategyproofness. Our impossibility is obtained using a formulation of the problem in propositional logic and applying SAT solvers; a human-readable version of the computer-generated proof is obtained by extracting a minimal unsatisfiable set (MUS). We also discuss several related axiomatic questions in the domain of committee elections.
\end{abstract}

\maketitle

\section*{Version Notes}

\indent\indent
[September 2020] This version is a slight revision of the paper that appeared in the conference proceedings. The conference version contained a gap in the proof of Lemma~\ref{lem:induction-m}. To fix this gap, it was necessary to introduce an additional axiom, weak efficiency, to the impossibility theorems; see Section~\ref{sec:efficiency}. I thank Boas Kluiving, Adriaan de Vries, Pepijn Vrijbergen for pointing out the error.

[November 2024] I removed a mistaken claim at the end of Section~\ref{sec:sp-axioms} that ``one can check that PAV cannot be manipulated by reporting a superset of one’s ballot''. PAV can be so manipulated.

\iflatexml
Version v3, November 2024. \\ Original version published at AAMAS 2018. \\ Current affiliation: CNRS, LAMSADE, Université Paris Dauphine - PSL. \\ Email: dominik.peters@lamsade.dauphine.fr
\fi

\section{Introduction}

The theory of multiwinner elections is concerned with designing and analysing procedures that, given preference information from a collection of voters, select a fixed-size \emph{committee} consisting of $k$ members, drawn from a larger set of $m$ candidates. Often, we will be interested in picking a \emph{representative} committee whose members together cover the diverse interests of the voters. We may also aim for this representation to be \emph{proportional}; for example, if a group of 20\% of the voters have similar interests, then about 20\% of the members of the committee should represent those voters' interests.

Historically, much work in mathematical social science has tried to formalise the latter type of proportionality requirement, in the form of finding solutions to the \emph{apportionment problem}, which arises in settings where voters express preferences over \emph{parties} which are comprised of many candidates \citep{BaYo82a}. 
More recently, theorists have focussed on cases where there are no parties, and preferences are expressed directly over the candidates \citep{FSST-trends}. The latter setting allows for applications in areas outside the political sphere, such as in group recommendation systems.

To formalise the requirement of proportionality in this party-free setting, it is convenient to consider the case where input preferences are given as \emph{approval ballots}: each voter reports a set of candidates that they find acceptable. Even for this simple setting, there is a rich variety of rules that exhibit different behaviour \citep{kil-handbook}, and this setting gives rise to a rich variety of axioms.

One natural way of selecting a committee of $k$ candidates when given approval ballots is to extend \emph{Approval Voting} (AV): for each of the $m$ candidates, count how many voters approve them (their \emph{approval score}), and then return the committee consisting of the $k$ candidates whose approval score is highest. Notably, this rule can produce committees that fail to represent large groups of voters. Consider, for example, an instance where $k=3$, and where 5 voters approve candidates $a$, $b$ and $c$, while 4 other voters approve only the candidate $d$. Then AV would select the committee $\{a,b,c\}$, leaving almost half of the electorate unrepresented. Intuitively, the latter group of 4 voters, consisting of more than a third of the electorate, should be represented by at least 1 of the 3 committee members.

\citet{ejr} introduce an axiom called \emph{justified representation} (JR) which formalises this intuition that a group of $n/k$ voters should not be left without any representation; a stronger version of this axiom called \emph{proportional justified representation} (PJR) has also been introduced and studied \citep{pjr17}. While AV fails these axioms, there are appealing rules which satisfy them. An example is \emph{Proportional Approval Voting} (PAV), first proposed by \citet{Thie95a}. The intuition behind this rule is that voters prefer committees which contain more of their approved candidates, but that there are decreasing marginal returns; specifically, let us presume that voters gain 1 `util' in committees that contain exactly 1 approved candidates, $1 + \frac12$ utils with 2 approved candidates, and in general $1 + \frac12 + \dots + \frac1r$ utils with $r$ approved candidates. PAV returns the committee that maximises utilitarian social welfare with this choice of utility function. PAV satisfies a strong form of justified representation \citep{ejr}.

When voters are strategic, PAV has the drawback that it can often be manipulated. Indeed, suppose a voter $i$ approves candidates $a$ and $b$. If $a$ is also approved by many other voters, PAV is likely to include $a$ in its selected committee anyway, but it might not include $b$ because voter $i$ is already happy enough due to the inclusion of $a$. However, if voter $i$ pretends not to approve $a$, then it may be utility-maximising for PAV to include both $a$ and $b$, so that $i$ successfully manipulated the election.%
\footnote{For a specific example, consider $P = (abc,abc,abc,abd,abd)$ for which $abc$ is the unique PAV-committee for $k=3$. If the last voter instead reports to approve $d$ only, then the unique PAV-committee is $abd$.}
Besides PAV, there exist several other proportional rules, such as rules proposed by Phragm\'en \citep{Janson16arxiv,brillphragmen}, but all of them can be manipulated using a similar strategy.

That voting rules are manipulable is very familiar to voting theorists; indeed the Gibbard--Satterthwaite theorem shows that for single-winner voting rules and strict preferences, \emph{every} non-trivial voting rule is manipulable.
However, in the approval-based multiwinner election setting, we have the tantalising example of Approval Voting (AV): this rule is strategyproof in the sense that voters cannot induce AV to return a committee including more approved candidates by misrepresenting their approval set. 
This raises the natural question of whether there exist committee rules that combine the benefits of AV and PAV: are there rules that are simultaneously proportional and strategyproof?

The contribution of this paper is to show that these two demands are incompatible. No approval-based multiwinner rule satisfies both requirements. This impossibility holds even for very weak versions of proportionality and of strategyproofness. The version of proportionality we use is much weaker than JR. It requires that if there is a group of at least $n/k$ voters who all approve a certain candidate $c$, and none of them approve any other candidate, and no other voters approve $c$, then $c$ should be part of the committee. Strategyproofness requires that a voter cannot manipulate the committee rule by dropping candidates from their approval ballot; a manipulation would be deemed successful if the voter ends up with a committee that contains additional approved candidates. In particular, our notion of strategyproofness only requires that the committee rule be robust to \emph{dropping} candidates; we do not require robustness against arbitrary manipulations that both add and remove candidates. Additionally, we impose a mild efficiency axiom requiring that the rule not elect candidates who are approved by none of the voters.

The impossibility theorem is obtained using computer-aided techniques that have recently found success in many areas of social choice theory \citep{GePe17a}. We encode the problem of finding a committee rule satisfying our axioms into propositional logic, and then use a SAT solver to check whether the formula is satisfiable. If the formula is unsatisfiable, this implies an impossibility, for a fixed number of voters, a fixed number of candidates, and a fixed $k$. We can then manually prove induction steps showing that the impossibility continues to hold for larger parameter values. Such techniques were first used by \citet{TaLi09a} to give alternative proofs of Arrow's and other classic impossibilities, and by \citet{GeEn11a} to find impossibilities for set extensions. \citet{BrGe15a} developed a method based on \emph{minimal unsatisfiable sets} that allows extracting a \emph{human-readable} proof of the base case impossibility. Thus, even though parts of the proofs in this paper are computer-generated, they are entirely human-checkable.

We begin our paper by describing several possible versions of strategyproofness and proportionality axioms. We then explain the computer-aided method for obtaining impossibility results in more detail, and present the proof of our main theorem. We end by discussing some extensions to this result, and contrast our result to a related impossibility theorem due to \citet{duddy}.

\section{Preliminaries}
\label{sec:prelims}

Let $C$ be a fixed finite set of $m$ \emph{candidates}, and let $N = \{1,\dots, n\}$ be a fixed finite set of $n$ \emph{voters}. An \emph{approval ballot} is a proper%
\footnote{Nothing hinges on the assumption that ballots are \emph{proper} subsets. Since we are mainly interested in impossibilities, this `domain restriction' slightly strengthens the results.}
subset $A_i$ of $C$, so that $\emptyset \neq A_i \subsetneq C$; let $\mathcal B$ denote the set of all ballots. For brevity, when writing ballots, we often omit braces and commas, so that the ballot $\{a,b\}$ is written $ab$. An (approval) \emph{profile} is a function $P : N \to \mathcal B$ assigning every voter an approval ballot. For brevity, we write a profile $P$ as an $n$-tuple, so that $P = (P(1),\dots,P(n))$. For example, in the profile $(ab,abc,d)$, voter $1$ approves candidates $a$ and $b$, voter $2$ approves $a$, $b$, and $c$, and voter $3$ approves $d$ only.

Let $k$ be a fixed integer with $1 \le k \le m$. A \emph{committee} is a subset of~$C$ of cardinality~$k$. We write $\mathcal C_k$ for the set of committees, and again for brevity, the committee $\{a,b\}$ is written as~$ab$.
An (approval-based) \emph{committee rule} is a function $f : \mathcal B^N \to \mathcal C_k$, assigning to each approval profile a unique winning committee. Note that this definition assumes that $f$ is \emph{resolute}, so that for every possible profile, it returns exactly one committee. In our proofs, we will implicitly restrict the domain of $f$ to profiles $P$ with $|\bigcup_{i\in N} P(i)| \ge k$, so that it is possible to fill the committee with candidates who are each approved by at least one voter. Since we are aiming for a negative result, this domain restriction only makes the result stronger.

Let us define two specific committee rules which will be useful examples throughout.

\emph{Approval Voting (AV)} is the rule that selects the $k$ candidates with highest approval score, that is, the $k$ candidates $c$ for which $|\{i\in N : c \in P(i) \}|$ is highest. Ties are broken lexicographically.

\emph{Proportional Approval Voting (PAV)} is the rule that returns the set $W\subseteq C$ with $|W| = k$ which maximises
\[ \sum_{i\in N} \left( 1 + \frac12 + \cdots + \frac1{|P(i) \cap W|} \right). \]
In case of ties, PAV returns the lexicographically first optimum.

Other important examples that we occasionally mention are Monroe's rule, Chamberlin--Courant, Phragm\'en's rules, and the sequential version of PAV. For definitions of these rules, we refer to the book chapter by \citet{FSST-trends}; they are not essential for following our technical results.

\section{Our Axioms}

In this section, we discuss the axioms that will be used in our impossibility result. These axioms have been chosen to be as weak as possible while still yielding an impossibility. This can make them sound technical and unnatural in isolation. To better motivate them, we discuss stronger versions that may have more natural appeal.

\subsection{Strategyproofness}
\label{sec:sp-axioms}
A voter can \emph{manipulate} a voting rule if, by submitting a non-truthful ballot, the voter can ensure that the voting rule returns an outcome that the voter strictly prefers to the outcome at the truthful profile. It is not obvious how to phrase this definition for committee rules, since we do not assume that voters have preferences over committees; we only have approval ballots over candidates.

One way to define manipulability in this context is to \emph{extend} the preference information we have to preferences over committees. This is the approach also typically taken when studying set-valued (irresolute) voting rules \citep{Tayl05a,Gard79a} or probabilistic voting rules \citep{Bran17a}. 
In our setting, there are several ways to extend approval ballots to preferences over committees, and hence several notions of strategyproofness. Our impossibility result uses the weakest notion.

For the formal definitions, let us introduce the notion of \emph{$i$-variants}.
For a voter $i\in N$, we say that a profile $P'$ is an $i$-variant of profile $P$ if $P$ and $P'$ differ only in the ballot of voter $i$, that is, if $P(j) = P'(j)$ for all $j \in N \setminus \{i\}$.
Thus, $P'$ is obtained after $i$ manipulated in some way, assuming that $P$ was the truthful profile.

One obvious way in which one committee can be better than another in a voter's view is if the former contains a larger number of approved candidates. Suppose at the truthful profile, we elect a committee of size $k=5$, of which voter $i$ approves 2 candidates. If $i$ can submit a non-truthful approval ballots which leads to the election of a committee with 3 candidates who are approved by $i$, then this manipulation would be successful in the cardinality sense.
\begin{description}
	\item[Cardinality-Strategyproofness] If $P'$ is an $i$-variant of $P$, then we do not have $|f(P') \cap P(i)| > |f(P) \cap P(i)|$.
\end{description}
One can check that AV with lexicographic tie-breaking satisfies cardinality-strategyproofness: it is neither advantageous to increase the approval score of a non-approved candidate, nor to decrease the approval score of an approved candidate.

Alternatively, we can interpret an approval ballot $A\in \mathcal B$ to say that the voter likes the candidates in $A$ (and would like them to in the committee), and that the voter dislikes the candidates not in $A$ (and would like them not to be in the committee). The voter's `utility' derived from committee $W$ would be the number of approved candidates in $W$ plus the number of non-approved candidates not in $W$. Interpreting approval ballots and committees as bit strings of length $m$, the voter thus desires the \emph{Hamming distance} between their ballot and the committee to be small. For two sets $A,B$, write $\mathcal H(A,B) = |A \mathbin{\Delta} B| = |(A\cup B) \setminus (A\cap B)|$.
\begin{description}
	\item[Hamming-Strategyproofness] If $P'$ is an $i$-variant of $P$, then we do not have $\mathcal H(f(P'), P(i)) < \mathcal H(f(P), P(i))$.
\end{description}
One can check that Hamming-strategyproofness and cardinality-strategyproofness are equivalent, because for a \emph{fixed} ballot $P(i)$, a committee is Hamming-closer to $P(i)$ than another if and only if the number of approved candidates is higher in the former.

The notions of strategyproofness described so far make sense if we subscribe to the interpretation of an approval ballot as a dichotomous preference, with the voter being completely indifferent between all approved candidates (or being unable to distinguish between them). In some settings, this is not a reasonable assumption. 

For example, suppose $i$ approves $\{a,b,c\}$; still it might be reasonable for $i$ to prefer a committee containing just $a$ to a committee containing both $b$ and $c$, maybe because $i$'s underlying preferences are such that $a$ is preferred to $b$ and $c$, even though all three are approved. However, $i$ should definitely prefer a committee that includes a strict superset of approved candidates. For example, a committee containing $a$ and $b$ should be better than a committee containing only $a$. This is the intuition behind superset-strategyproofness, which is a weaker notion than cardinality-strategyproofness.
\begin{description}
	\item[Superset-Strategyproofness] If $P'$ is an $i$-variant of $P$, then we do not have $f(P') \cap P(i) \supsetneq f(P) \cap P(i)$.
\end{description}

Interestingly, PAV and other proportional rules are often manipulable in a particularly simple fashion: a manipulator can obtain a better outcome by dropping popular candidates from their approval ballot. Formally, these rules can be manipulated even through reporting a proper subset of the truthful ballot. Our final and official notion of strategyproofness is a version of subset-strategyproofness which only requires the committee rule to resist manipulators who report a subset of the truthful ballot.

\begin{description}
	\item[Strategyproofness] If $P'$ is an $i$-variant of $P$ with $P'(i) \subset P(i)$, then we do not have $f(P') \cap P(i) \supsetneq f(P) \cap P(i)$.
\end{description}

Manipulating by reporting a subset of one's truthful ballot is sometimes known as \emph{Hylland free riding} \citep{hylland1992,schulze2004freeriding}: the manipulator free-rides on others approving a candidate, and can pretend to be worse off than they actually are. This can then induce the committee rule to add further candidates from their ballot to the committee. 
\citet{aziz2017fair} study a related notion of `excludable strategyproofness'
in the context of probabilistic voting rules.

\vspace{\baselineskip}
\vspace{\baselineskip}

\subsection{Proportionality}
\label{sec:prop-axioms}

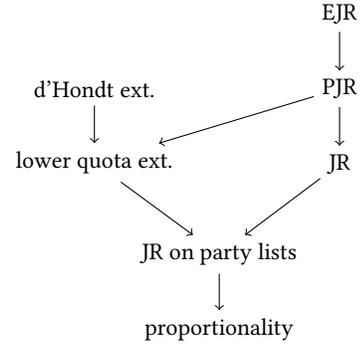
\begin{figure}
	\begin{tikzpicture}
		\node (EJR) {EJR};
		\node [below of=EJR] (PJR) {PJR};
		\node [below of=PJR] (JR) {JR};
		
		\node [left=2cm of PJR] (dH) {d'Hondt ext.};
		\node [below=0.5cm of dH] (lq) {lower quota ext.};
		
		\node [below left=0.7cm and 0.2cm of JR] (JRp) {JR on party lists};
		\node [below of=JRp] (prop) {proportionality};
		
		\draw[->]
		  (EJR) edge (PJR)
		  (PJR) edge (JR)
		  (JR) edge (JRp)
		  (JRp) edge (prop)
		  (PJR) edge (lq)
		  (dH) edge (lq)
		  (lq) edge (JRp);		
	\end{tikzpicture}
	\caption{Proportionality axioms and logical implications.}
	\label{fig:prop}
\end{figure}

We now discuss several axioms formalising the notion that the committee rule $f$ should be \emph{proportional}, in the sense of proportionally representing different factions of voters: for example, a `cohesive' group of $10\%$ of the voters should be represented by about $10\%$ of the members of the committee. The version of proportionality used in our impossibility is the last axiom we discuss. All other versions imply the one leading to impossibility; thus, this version is the weakest notion among the ones discussed here. Figure~\ref{fig:prop} shows a Hasse diagram of all discussed axioms. Approval Voting (AV) fails all of them, as can be checked for the example profile $P = (abc,abc,d)$ and $k = 3$, where AV returns $abc$.

We say that a profile $P$ is a \emph{party-list profile} if for all voters $i,j\in N$, either $P(i) = P(j)$, or $P(i) \cap P(j) = \emptyset$.
For example, $(ab,ab,cde,cde,\textit{f})$ is a party-list profile, but $(ab,c,c,abc)$ is not.
A party-list profile induces a partition of the set $C$ of candidates into disjoint \emph{parties}, so that each voter approves precisely the members of exactly one party.
The problem of finding a proportional committee given a party-list profile has been extensively studied as the problem of \emph{apportionment}.
Functions $g : \{ \text{party-list profiles}  \} \to \mathcal C_k$ are known as \emph{apportionment methods}; thus any committee rule induces an apportionment by restricting its domain to party-list profiles \citep{BLS17a}.
Many proportional apportionment methods have been introduced and defended over the last few centuries.
Given a committee rule $f$, one way to formalise the notion that $f$ is proportional is by requiring that the apportionment method induced by $f$ is proportional.

Given a party-list profile $P$, let us write $n_P(A) = |\{i \in N: P(i) = A\}|$ for the number of voters approving party $A$.
An apportionment method $g$ satisfies \emph{lower quota} if for every party-list profile $P$, each party $A$ in $P$ gets at least $\lfloor n_P(A) \cdot \frac k n \rfloor$ seats, that is, $|g(P) \cap A| \ge \lfloor n_P(A) \cdot \frac k n \rfloor$. 
This notion gives us our first proportionality axiom.
\begin{description}
	\item[Lower quota extension] The apportionment method induced by $f$ satisfies lower quota.
\end{description}
This axiom is satisfied by PAV, the sequential version of PAV, by Monroe's rule if $k$ divides $n$, and Phragm\'en's rule \citep{BLS17a}.

We can strengthen this axiom by imposing stronger conditions on the induced apportionment method. For example, the apportionment method induced by PAV and by Phragm\'en's rule  coincides with the \emph{d'Hondt method} (aka Jefferson method, see \citep{BLS17a} for a definition), so we could use the following axiom.
\begin{description}
	\item[d'Hondt extension] The apportionment method induced by $f$ is the d'Hondt method.
\end{description}

\citet{ejr} introduce a different approach of defining a proportionality axiom. Instead of considering only the case of party-list profiles, they impose conditions on \emph{all} profiles. The intuition behind their axioms is that sufficiently large groups of voters that have similar preferences `deserve' at least a certain number of representatives in the committee. They introduce the following axiom:
\begin{description}
	\item[Justified Representation (JR)] If $P$ is a profile, and $N' \subseteq N$ is a group with $|N'| \ge \frac nk$ and $\bigcap_{i\in N'} P(i) \neq \emptyset$, then $f(P) \cap \bigcup_{i\in N'} P(i) \neq \emptyset$.
\end{description}
Thus, JR requires that no group of at least $\frac nk$ voters for which there is at least one candidate $c \in C$ that they all approve can remain unrepresented: at least one of the voters in the group must approve at least one of the committee members.
This axiom is satisfied, for example, by PAV, Phragm\'en's rule, and Chamberlin-Courant \citep{ejr}, but not by the sequential version of PAV unless $k \le 5$ \citep{pjr17,ejr}.

One may think that JR is too weak: even if there is a large majority of voters who all report the same approval set, JR only requires that \emph{one} of their candidates be a member of the committee. But this group may deserve several representatives. The following strengthened version of JR is due to \citep{pjr17}. It requires that a large group of voters for which there are several candidates that they all approve should be represented by several committee members.
\begin{description}
	\item[Proportional Justified Representation (PJR)] For any profile $P$ and each $\ell = 1,\dots,k$, if $N' \subseteq N$ is a group with $|N'| \ge \ell\cdot\frac nk$ and $|\bigcap_{i\in N'} P(i)| \ge \ell$, then $|f(P) \cap \bigcup_{i\in N'} P(i)| \ge \ell$.
\end{description}
This axiom is also satisfied by PAV and Phragm\'en's rule \citep{pjr17,brillphragmen}. \citet{BLS17a} show that if a rule satisfies PJR, then it is also a lower quota extension. A yet stronger version of JR is \emph{EJR}, introduced by \citet{ejr}; EJR requires that there is at least one group member who has at least $\ell$ approved committee members.
\begin{description}
	\item[Extended Justified Representation (EJR)] For any profile $P$ and each $\ell = 1,\dots,k$, if $N' \subseteq N$ is a group with $|N'| \ge \ell\cdot\frac nk$ and $|\bigcap_{i\in N'} P(i)| \ge \ell$, then $|f(P) \cap P(i)| \ge \ell$ for some $i\in N'$.
\end{description}
This axiom is satisfied by PAV \citep{ejr}, but not by Phragm\'en's rule \citep{brillphragmen}.

The proportionality axiom we use in our impossibility combines features of the JR-style axioms with the apportionment-extension axioms. Consider the following axiom.
\begin{description}
	\item[JR on party lists] Suppose $P$ is a party-list profile, and some ballot $A \in \mathcal B$ appears at least $\frac n k$ times in $P$. Then $f(P) \cap A \neq \emptyset$.
\end{description}
This axiom only requires JR to hold for party-list profiles; thus, it only requires that we represent large-enough groups of voters who all report the exact same approval ballot \citep[see also][]{BKNS14a}.
As an example, this axiom requires that $f(ab, ab, cd, cd) \in \{ac, ad, bc, bd\}$, because the ballots $ab$ and $cd$ both appear at least $\frac nk = \frac 42 = 2$ times.

Our official proportionality axiom is still weaker, and only requires us to represent \emph{singleton} parties with large-enough support.
\begin{description}
	\item[Proportionality] Suppose $P$ is a party-list profile, and some singleton ballot $\{c\} \in \mathcal B$ appears at least $\frac n k$ times in $P$. Then $c \in f(P)$.
\end{description}
This axiom should be almost uncontroversial if we desire our committee rule to be proportional in any sense. A group of voters who all approve just a single candidate is certainly cohesive (there are \emph{no} internal disagreements), it is clear what it means to represent this group (add their approved candidate to the committee), and the group is uniquely identified (because no outside voters approve sets that intersect with the group's approval ballot).

Since our proportionality axiom only refers to the apportionment method induced by $f$, our impossibility states that no reasonable apportionment method admits an extension to the `open list' setting (where voters are not bound to a party) which is strategyproof.

A type of axiom related to proportionality are \emph{diversity} requirements. These typically require that as many voters as possible should have a representative in the committee, but they do not insist that groups of voters be proportionally represented \citep{elk-fal-sko-sli:c:multiwinner-rules,FSST-trends}. The Chamberlin--Courant rule \citep{ChCo83a} is an example of a rule selecting diverse committees. \citet{lac-sko:t:abc-approval-multiwinner} propose the following formulation of this requirement for the approval setting:
\begin{description}
	\item[Disjoint Diversity] Suppose $P$ is a party-list profile with at most $k$ different parties. Then $f(P)$ contains at least one member from each party.
\end{description}
Our main result (Theorem~\ref{thm:main}) also holds when replacing proportionality by disjoint diversity, since all profiles in its proof where proportionality is invoked feature at most $k$ different parties.

\subsection{Efficiency}
\label{sec:efficiency}

We will additionally impose a mild technical condition, which can be seen as an efficiency axiom.%
The axiom will only be used in one of the induction steps (Lemma~\ref{lem:induction-m}).

\begin{description}
	\item[Weak Efficiency] If $P$ is a profile with $|\bigcup_{i\in N} P(i)| \ge k$, and $c$ is a candidate who is approved by no voters, then $c \not\in f(P)$.
\end{description}

Thus, a rule satisfying weak efficiency should fill the committee with candidates who are approved by some voters, rather than electing candidates approved by no one. A similar axiom of the same name is used by \citet{lac-sko:t:abc-approval-multiwinner}. As we declared in Section~\ref{sec:prelims}, in our proofs we will always restrict attention to profiles $P$ with $|\bigcup_{i\in N} P(i)| \ge k$, so that weak efficiency applies to all relevant profiles.

\section{The Computer-Aided Approach}

To obtain our impossibility result, we have used the computer-aided technique developed by \citet{TaLi09a} and \citet{GeEn11a}. This approach is based on using a computer search (usually in form of a SAT solver) to establish the base case of an impossibility theorem, and then using (manually proved) induction steps to extend the theorem to bigger values of $n$ and $m$. \citet{TaLi09a} used this technique to give proofs of Arrow's and other classic impossibility theorems in social choice, and \citet{GeEn11a} used it to find new impossibilities in the area of set extensions. A paper by \citet{BrGe15a} used this approach to prove an impossibility about strategyproof \emph{tournament solutions}; an important technical contribution of their paper was the use of \emph{minimal unsatisfiable sets} to produce human-readable proofs of the base case. This technology was also used to prove new impossibilities about the no-show paradox \citep{BGP16c}, about half-way monotonicity \citep{peters2017halfway}, and about probabilistic voting rules \citep{BBG15a}. A recent book chapter by \citet{GePe17a} provides a survey of these results.

The ``base case'' of an impossibility theorem proves that no voting rule exists satisfying a certain collection of axioms for a \emph{fixed} number of voters and alternatives, and (in our case) a fixed committee size $k$. Fixing these numbers, there are only finitely many possible rules, and we can in principle iterate through all $\binom{m}{k}^{\raisebox{-2pt}{$\scriptstyle2^{mn}$}}$possibilities and check whether any satisfies our axioms. However, this search space quickly grows out of reach of a na\"ive search. 

In many cases, we can specify our axiomatic requirements in propositional logic, and use a SAT solver to check for the existence of a suitable voting rule. Due to recent dramatic improvements in solving times of SAT solvers, this approach often makes this search feasible, even for moderately large values of $n$ and $m$ \citep{BGP16c}.

How can we encode our problem of finding a proportional and strategyproof committee rule into propositional logic? This turns out to be straightforward. Our formula will be specified so that every satisfying assignment explicitly encodes a committee rule satisfying the axioms. We generate a list of all $(2^m)^n$ possible approval profiles, and for each profile $P$ and each committee $W$, we introduce a propositional variable $x_{P,W}$ with the intended interpretation that
\[ x_{P,W} \text{ is true} \iff f(P) = W. \]
We then add clauses that ensure that any satisfying assignment encodes a function (so that $f(P)$ takes exactly one value), we add clauses that ensure that only proportional committees may be returned, and we iterate through all profiles $P$ and all $i$-variants of it, adding clauses to ensure that no successful manipulations are possible. The details are shown in Algorithm~\ref{alg:one}; the formulation we use is slightly more efficient by never introducing the variable $x_{P,W}$ in case that the committee $W$ is not proportional in profile $P$.

Now, given numbers $n$, $m$, and $k$, Algorithm~\ref{alg:one} encodes our problem, passes the resulting propositional formula to a SAT solver \citep{Bier13a,AuSi09a} and reports whether the formula was satisfiable. If it is satisfiable, then we know that there exists a propotional and strategyproof committee rule for these parameter values, and the SAT solver will return an explicit example of such a rule in form of a look-up table. If the formula is unsatisfiable (like in our case), then we have an impossibility for these parameter values.

A remaining challenge is to extend this impossibility result to other parameter values, which is usually done by proving \emph{induction steps}; however, this is not always straightforward to do, and in some cases, impossibilities do \emph{not} hold for all larger parameter values \citep[e.g.,][]{peters2017halfway}. In many cases, the induction step on $n$ is most-difficult to establish. We also run into trouble proving this step, and our impossibility is only proved for the case where $n$ is a multiple of $k$.

Another challenge is to find a proof of the obtained impossibility, and preferably one that can easily be checked by a human. Many SAT solvers can be configured to output a proof trace which contains all steps used to deduce that the formula is unsatisfiable; but these proofs can become very large. Recent examples are SAT-generated proofs of a special case of the Erd\H{o}s Discrepancy Conjecture \citep{KoLi14a} which takes 13GB, and of a solution to the Boolean Pythagorean Triples Problem \citep{HKM16a} which takes 200TB. Clearly, humans cannot check the correctness of these proofs.

We use the method introduced by \citet{BrGe15a} via minimal unsatisfiable sets (MUS). An MUS of an unsatisfiable propositional formula in conjunctive normal form is a subset of its clauses which is already unsatisfiable, but minimally so: removing any further clause leaves a satisfiable formula. Thus, every clause in an MUS corresponds to a `proof ingredient' which cannot be skipped. MUSes of formulas derived from voting problems like ours are often very small, only referring to a few dozen profiles. This can be explained through the `local' nature of the axioms used: proportionality constrains the behaviour of the committee rule at a single profile, and strategyproofness links the behaviour at two profiles.

MUSes can be found using MUS extractors, which have become reasonably efficient. We used MUSer2 \citep{BeMa12a} and MARCO \citep{LPMM15a}. Once one finds a small MUS, it can then be manually inspected to understand how the clauses in the MUS fit together. More details of this process are described in the book chapter by \citet{GePe17a}.

\begin{algorithm}[t]
	\DontPrintSemicolon
	\SetAlgoNoLine
	\KwIn{Set $C$ of candidates, set $N$ of voters, committee size $k$.}
	\KwOut{Does a proportional and strategyproof committee rule exist?}
	\For{each profile $P \in \mathcal B^N$}{
		\eIf{$P$ is a party-list profile}{
			allowed$[P] \gets \{ C \in \mathcal C_k : C \text{ provides JR to singleton parties}  \}$\;
		}{
			allowed$[P] \gets \mathcal C_k$\;
		}
		\For{each committee $C \in \textup{allowed}[P]$}{
			introduce propositional variable $x_{P,C}$\;
		}
	}
	\For{each profile $P \in \mathcal B^N$}{
		add clause $\biglor_{C\in \textup{allowed}[P]} x_{P,C}$\;
		add clauses $\bigland_{C\neq C'\in \textup{allowed}[P]} (\lnot x_{P,C} \lor \lnot x_{P,C'})$\;
		\For{each voter $i\in N$}{
			\For{each $i$-variant $P'$ of $P$ with $P'(i) \subseteq P(i)$}{
				\For{each $C \in \textup{allowed}[P]$ and $C' \in \textup{allowed}[P']$}{
					\If{$C' \cap P(i) \supsetneq C \cap P(i)$}{
							add clause $(\lnot x_{P,C} \lor \lnot x_{P', C'})$\;
					}
				}
			}
		}
	}
	pass formula to SAT solver\;
	\Return{\textup{whether formula is satisfiable}}
	\caption{Encode Problem for SAT Solving}
	\label{alg:one}
\end{algorithm}

\section{The Impossibility Theorem}

We are now in a position to state our main result, that there are no proportional and strategyproof committee rules.

\begin{theorem}
	\label{thm:main}
	Suppose $k \ge 3$, the number $n$ of voters is divisible by $k$, and $m \ge k+1$. Then there exists no approval-based committee rule which satisfies weak efficiency, proportionality and strategyproofness.
\end{theorem}

The assumption that $k\ge 3$ is critical; we discuss the cases $k = 1$ and $k=2$ separately in Section~\ref{sec:k-2}. The assumption that $n$ be divisible by $k$ also appears to be critical; the SAT solver indicates positive results when $n$ is not a multiple of $k$. However, we do not know short descriptions of these rules, and it is possible (likely?) that impossibility holds for large $n$ and $m$. Using stronger proportionality axioms, the result holds for all sufficiently large $n$; see Section~\ref{sec:droop}.

The proof of this impossibility was found with the help of computers, but it was significantly simplified manually. One convenient first step is to establish the following simple lemma. It uses strategyproofness to extend the applicability of proportionality to certain profiles that are not party-list profiles.

\begin{lemma}
	\label{lem:singleton-approvers}
	Let $m = k+1$. Let $f$ be strategyproof and proportional. Suppose that $P$ is a profile in which some singleton ballot $\{c\}$ appears at least $\frac nk$ times, but in which no other voter approves $c$. Then $c\in f(P)$.
\end{lemma}
\begin{proof}
	Let $P'$ be the profile defined by
	\[ P'(i) = \begin{cases}
	\{c\} & \text{if } P(i) = \{c\}, \\
	C \setminus \{c\} & \text{otherwise}.
	\end{cases}  \]
	Then $P'$ is a party-list profile, and by proportionality, $c\in f(P')$. Thus, $f(P') \neq C \setminus \{c\}$. Now, step by step, we let each non-$\{c\}$ voter $j$ in $P'$ change back their vote to $P(j)$. By strategyproofness, at each step the output committee cannot be $C \setminus \{c\}$. In particular, at the last step, we have $f(P) \neq C \setminus \{c\}$. Thus, $c\in f(P)$, as required.
\end{proof}

\subsection{Base case}
\label{sec:base-case}

The first major step in the proof is to establish the impossibility in the case that $k = 3$, $n = 3$, and $m = 4$. 
The proof of this base case is by contradiction, assuming there exists some $f$ satisfying the axioms. We start by considering the profile $P_1 = (ab,c,d)$, and break some symmetries. (This is a useful strategy to obtain smaller and better-behaved MUSes.) Using proportionality, symmetry-breaking allows us to assume that $f(P_1) = acd$. The proof then goes through seven steps, applying the same reasoning each time. In each step, we use strategyproofness to infer the values of $f$ at certain profiles $P_2, \dots, P_7$.
Finally, we find that strategyproofness implies that $f(P_1) \neq acd$, which contradicts our initial assumption about $f(P_1)$. 

\begin{lemma}
	\label{lem:base-case}
	There is no committee rule that satisfies proportionality and strategyproofness for $k=3$, $n=3$, and $m=4$.
\end{lemma}
\begin{proof}
	Suppose for a contradiction that such a committee rule $f$ existed.
	Consider the profile $P_1 = (ab,c,d)$. By proportionality, we have $c\in f(P_1)$ and $d \in f(P_1)$. Thus, we have $f(P_1) \in \{ acd, bcd \}$. By relabelling the alternatives, we may assume without loss of generality that $f(P_1) = acd$.
	
	Consider $P_{1.5} = (ab,ac,d)$. By Lemma~\ref{lem:singleton-approvers}, $d \in f(P_{1.5})$. Thus, $f(P_{1.5}) = acd$, or else voter 2 can manipulate towards $P_1$.
	
	Consider $P_2 = (b,ac,d)$. By proportionality, $f(P_2) \in \{abd,bcd\}$. If we had $f(P_2) = abd$, then voter 1 in $P_{1.5}$ could manipulate towards $P_2$. Hence $f(P_2) = bcd$.
	
	Consider $P_{2.5} = (b,ac,cd)$. By Lemma~\ref{lem:singleton-approvers}, $b \in f(P_{2.5})$. Thus, $f(P_{2.5}) = bcd$, or else voter 3 can manipulate towards $P_2$.
	
	Consider $P_3 = (b,a,cd)$. By proportionality, $f(P_3) \in \{abc,abd\}$. If we had $f(P_3) = abc$, then voter 2 in $P_{2.5}$ could manipulate towards $P_3$. Hence $f(P_3) = abd$.
	
	Consider $P_{3.5} = (b,ad,cd)$. By Lemma~\ref{lem:singleton-approvers}, $b \in f(P_{3.5})$. Thus, $f(P_{3.5}) = abd$, or else voter 2 can manipulate towards $P_3$.
	
	Consider $P_4 = (b,ad,c)$. By proportionality, $f(P_4) \in \{abc,bcd\}$. If we had $f(P_4) = bcd$, then voter 3 in $P_{3.5}$ could manipulate towards $P_4$. Hence $f(P_4) = abc$.
	
	Consider $P_{4.5} = (b,ad,ac)$. By Lemma~\ref{lem:singleton-approvers}, $b \in f(P_{4.5})$. Thus, $f(P_{4.5}) = abc$, or else voter 3 can manipulate towards $P_4$.
	
	Consider $P_5 = (b,d,ac)$. By proportionality, $f(P_5) \in \{abd,bcd\}$. If we had $f(P_5) = abd$, then voter 2 in $P_{4.5}$ could manipulate towards $P_5$. Hence $f(P_5) = bcd$.
	
	Consider $P_{5.5} = (b,cd,ac)$. By Lemma~\ref{lem:singleton-approvers}, $b \in f(P_{5.5})$. Thus, $f(P_{5.5}) = bcd$, or else voter 2 can manipulate towards $P_5$.
	
	Consider $P_6 = (b,cd,a)$. By proportionality, $f(P_6) \in \{abc,abd\}$. If we had $f(P_6) = abc$, then voter 3 in $P_{5.5}$ could manipulate towards $P_6$. Hence $f(P_6) = abd$.
	
	Consider $P_{6.5} = (b,cd,ad)$. By Lemma~\ref{lem:singleton-approvers}, $b \in f(P_{6.5})$. Thus, $f(P_{6.5}) = abd$, or else voter 3 can manipulate towards $P_6$.
	
	Consider $P_7 = (b,c,ad)$. By proportionality, $f(P_7) \in \{abc,bcd\}$. If we had $f(P_7) = bcd$, then voter 2 in $P_{6.5}$ could manipulate towards $P_7$. Hence $f(P_7) = abc$.
	
	Finally, consider $P_{7.5} = (ab,c,ad)$. By Lemma~\ref{lem:singleton-approvers}, $c \in f(P_{7.5})$. Thus, $f(P_{7.5}) = abc$, or else voter 1 can manipulate towards $P_7$. But then voter 3 can manipulate towards $P_1 = (ab,c,d)$, because by our initial assumption, we have $f(P_1) = acd$. Contradiction.
\end{proof}

\subsection{Induction steps}
\label{sec:induction-step}

We now extend the base case to larger parameter values, by proving induction steps. The proofs all take the same form: Assuming the existence of a committee rule satisfying the axioms for large parameter values, we construct a rule for smaller values, and show that the smaller rule inherits the axiomatic properties of the larger rule. This is done, variously, by introducing dummy voters, by introducing dummy alternatives, and by copying voters. 

Our first induction step reduces the number of voters. The underlying construction works by copying voters, and using the `homogeneity' of the axioms of proportionality and strategyproofness. For the latter axiom, we use the fact that in the case $m = k+1$, the preference extension of approval ballots to committees is \emph{complete}, in that any two committees are comparable.

\begin{lemma}
	\label{lem:induction-n}
	Suppose $k\ge 2$ and $m = k+1$, and let $q \ge 1$ be an integer. If there exists a proportional and strategyproof committee rule for $q\cdot k$ voters, then there also exists such a rule for $k$ voters.
\end{lemma}
\begin{proof}
	For convenience, we write profiles as lists. Given a profile $P$, we write $qP$ for the profile obtained by concatenating $q$ copies of $P$. Let $f_{qk}$ be the rule for $q\cdot k$ voters. We define the rule $f_k$ for $k$ voters as follows:
	\[ f_k(P) = f_{qk}(qP) \quad \text{for all profiles $P \in \mathcal B^k$.}  \]
	
	\emph{Proportionality.} Suppose $P \in \mathcal B^k$ is a party-list profile in which at least $\frac nk = \frac kk = 1$ voters approve $\{c\}$. Then $qP$ is a party-list profile in which at least $q \cdot \frac nk = \frac{qn}k = q$ voters approve $\{c\}$. Since $f_{qk}$ is proportional, $c\in f_{qk}(qP) = f_k(P)$.
	
	\emph{Strategyproofness.} Suppose for a contradiction that $f_k$ is not strategyproof, so that there is $P$ and an $i$-variant $P'$ with $f_k(P') \cap P(i) \supsetneq f_k(P) \cap P(i)$. Because $m = k+1$, the committees $f_k(P')$ and $f_k(P)$ must differ in exactly 1 candidate. Since the manipulation was successful, $f_k(P')$ must be obtained by replacing a non-approved candidate in $f_k(P)$ by an approved one, say $f_k(P') = f_k(P) \cup \{c\} \setminus \{d\}$ with $c \in P(i) \not\ni d$. Now consider $f_{qk}(qP)$, and step-by-step let each of the $q$ copies of $P(i)$ in $qP$ manipulate from $P(i)$ to $P'(i)$ obtaining $qP'$ in the last step. Because $f_{qk}$ is strategyproof, at each step of this process $f_{qk}$ cannot have exchanged a non-approved candidate by an approved candidate according to $P(i)$. This contradicts that $f_k(P') = f_k(P) \cup \{c\} \setminus \{d\}$.
\end{proof}

Our second induction step is the simplest: We reduce the number of alternatives using dummy candidates that no voter ever approves. This is the only place in the proof where we require the weak efficiency axiom.

\begin{lemma}
	\label{lem:induction-m}
	Fix $n$ and $k$, and let $m \ge k$. If there exists a weakly efficient, proportional, and strategyproof committee rule for $m+1$ alternatives, then there also exists such a rule for $m$ alternatives.
\end{lemma}
\begin{proof}
	Let $f_{m+1}$ be the committee rule defined on the candidate set $C_{m+1} = \{c_1, \dots, c_m, c_{m+1}\}$. Note that every profile $P$ over candidate set $C_m = \{ c_1, \dots, c_m \}$ is also a profile over candidate set $C_{m+1}$. We then just define the committee rule $f_m$ for the candidate set $C_m$ by $f_m(P) := f_{m+1}(P)$ for all profiles $P$ over candidate set $C_m$, where we assume that $|\bigcup_{i\in N} P(i)| \ge k$. By weak efficiency, $f_m(P) \subseteq C_m$, so that $f_m$ is a well-defined rule. It is easy to check that $f_m$ is weakly efficient, proportional, and strategyproof.
\end{proof}

Our last induction step reduces the committee size from $k+1$ to $k$. The construction introduces an additional candidate and an additional voter, and appeals to Lemma~\ref{lem:singleton-approvers} to show that the new candidate is always part of the winning committee. Thus, the larger rule implicitly contains a committee rule for size-$k$ committees.

\begin{lemma}
	\label{lem:induction-k}
	Let $k \ge 2$. If there exists a proportional and strategyproof committee rule for committee size $k+1$, for $k+1$ voters, and for $k+2$ alternatives, then there also exists such a rule for committee size $k$, for $k$ voters, and for $k+1$ alternatives.
\end{lemma}
\begin{proof}
	Let $f_{k+1}$ be the committee rule assumed to exist, defined on the candidate set $C_{k+2} = \{ c_1, \dots, c_{k+2} \}$. We define the rule $f_k$ for committee size $k$ on candidate set $C_{k+1} = \{c_1, \dots, c_{k+1}\}$ as follows:
	\[ f_k(A_1, \dots, A_k) = f_{k+1}(A_1, \dots, A_k, \{ c_{k+2} \}) \setminus \{ c_{k+2} \}, \]
	for every profile $P = (A_1, \dots, A_k)$ over $C_{k+1}$. Notice that this is well-defined and returns a committee of size $k$, since by Lemma~\ref{lem:singleton-approvers} applied to $f_{k+1}$, we always have $c_{k+2} \in f_{k+1}(A_1, \dots, A_k, \{ c_{k+2} \})$.
	
	\emph{Proportionality.} Let $P = (A_1, \dots, A_k)$ be a party-list profile over $C_{k+1}$, in which the ballot $\{c\}$ occurs at least $\frac nk = \frac kk = 1$ time. Then $P' = (A_1, \dots, A_k, \{ c_{k+2} \})$ is a party-list profile, in which $\{c\}$ occurs at least $\frac {n+1}{k+1} = \frac{k+1}{k+1} = 1$ time; thus, by proportionality of $f_{k+1}$, we have $c\in f_{k+1}(P') = f_k(P)$.
	
	\emph{Strategyproofness.} If there is a successful manipulation from $P$ to $P'$ for $f_k$, then there is a successful manipulation from $(P, \{c_{k+2}\})$ to $(P', \{c_{k+2}\})$ for $f_{k+1}$, contradiction.
\end{proof}

Finally, we can combine all three induction steps, applying them in order, and the base case, to get our main result.

\begin{proof}[Proof of the Main Theorem.]
	Let $k \ge 3$, let $n$ be divisible by $k$, and let $m \ge k+1$. Suppose for a contradiction that there does exist an approval-based committee rule $f$ which satisfies weak efficiency, proportionality, and strategyproofness for these parameters.
	
	By Lemma~\ref{lem:induction-m} applied repeatedly to $f$, there also exists such a rule $f'$ for $k+1$ alternatives. By Lemma~\ref{lem:induction-n} applied to $f'$, there exists a proportional and strategyproof rule $f''$ for $k$ voters. By Lemma~\ref{lem:induction-k} applied to $f''$, there must exist a proportional and strategyproof rule for committee size $3$, for $3$ voters, and for $4$ alternatives. But this contradicts Proposition~\ref{lem:base-case}.
\end{proof}

\subsection{Extension to other electorate sizes}
\label{sec:droop}

One drawback of Theorem~\ref{thm:main} is the condition on the number of voters $n$. For larger values of $k$, practical elections are unlikely to have a number of voters which is exactly a multiple of $k$. The impossibility as we have proved it does not rule out that for other values of $n$, there does exist a proportional and strategyproof rule. Indeed, at least for small parameter values, the SAT solver confirms that this is the case. An important open question is whether, for fixed $k\ge 3$, the impossibility holds for all sufficiently large $n$.

In this section, we give one result to this effect, obtained by strengthening the proportionality axiom. Note that all the axioms we discussed in Section~\ref{sec:prop-axioms} are based on the intuition that a group of $\frac nk$ voters should be represented by one committee member. The value ``$\frac nk$'' is known as the \emph{Hare quota}. An alternative proposal is the \emph{Droop quota}, according to which every group consisting of strictly more than $\frac n {k+1}$ voters should be represented by one committee member. Thus, with Droop quotas, slightly smaller groups already need to be represented. The strengthened axiom is as follows.

\begin{description}
	\item[Droop Proportionality] Suppose $P$ is any profile, and some singleton ballot $\{c\} \in \mathcal B$ appears strictly more than $\frac n {k+1}$ times in $P$. Then $c \in f(P)$.
\end{description}
Note that Droop proportionality applies to all profiles and not only party-list profiles.
With this stronger proportionality axiom, we can show that for fixed $k$ and \emph{all} sufficiently large $n$, we have an incompatibility with strategyproofness.

\begin{proposition}
	\label{prop:droop}
	Let $k \ge 3$, let $m \ge k+1$, and let $n \ge k^2$. Then there is no approval-based committee rule satisfying weak efficiency, strategyproofness, and Droop proportionality.
\end{proposition}
\begin{proof}
	Suppose such a rule $f_n$ exists. By Lemma~\ref{lem:induction-m} (suitably reproved to apply to the Droop quota), there also is such a rule for $m = k+1$ alternatives, so we may assume that $m = k+1$.
	
	Write $n = q\cdot k + r$ for some $0 \le r < k$ and some $q \ge k$. We will show that there exists a committee rule for $q\cdot k$ voters which satisfies proportionality (with respect to the \emph{Hare} quota) and strategyproofness, which contradicts Theorem~\ref{thm:main}.
	
	Fix $r$ arbitrary ballots $B_1, \dots, B_r$. We define a committee rule $f_{qk}$ on $q\cdot k$ voters, $m$ alternatives, and for committee size $k$, as follows:
	\[ f_{qk}(A_1, \dots, A_{qk}) = f_{n}(A_1, \dots, A_{qk}, B_1, \dots, B_r), \]
	for all profiles $P = (A_1, \dots, A_{qk}) \in \smash{\mathcal B^{qk}}$.
	
	It is clear that $f_{qk}$ inherits strategyproofness from $f_n$: Any successful manipulation of $f_{qk}$ is also successful for $f_n$.
	
	We are left to show that $f_{qk}$ satisfies (Hare) proportionality. So suppose that $P = (A_1, \dots, A_{qk}) \in \smash{\mathcal B^{qk}}$ is a party-list profile in which singleton party $\{c\}$ is approved by at least $\smash{\frac{qk}{k}} = q$ voters. 
	Note that, because $r < k \le q$,
	\[ \frac{n}{k+1} = \frac{qk+r}{k+1} < \frac{qk + q}{k+1} = \frac{q(k+1)}{k+1}  = q, \]
	Thus, in the profile $P' = (A_1, \dots, A_{qk}, B_1, \dots, B_r)$, there are strictly more than $\frac{n}{k+1}$ voters who approve $\{c\}$. Thus, by Droop proportionality, $c \in f_n(P') = f_{qk}(P)$. Thus, $f_{qk}$ is (Hare) proportional.
\end{proof}

\paragraph{Remark.}
If we want to restrict the Droop proportionality axiom to only apply to party-list profiles, we can instead assume in Proposition~\ref{prop:droop} that $m \ge k+2$, and then let $B_1 = \cdots = B_r = \{c_{k+2}\}$, defining the rule $f_{qk}$ only over the first $k+1$ alternatives. Then the final profile $P'$ is a party-list profile.

\subsection{Small committees}
\label{sec:k-2}
Theorem~\ref{thm:main} only applies to the case where $k \ge 3$.
For the case $k = 1$, where we elect just a single winner, Approval Voting with lexicographic tie-breaking is both proportional and strategyproof.%
\footnote{It is well-known that AV is strategyproof. Proportionality for $k=1$ is equivalent to a unanimity condition, since $\smash{\frac n k} = n$, and AV satisfies unanimity.}
This leaves open the case of $k = 2$.

The SAT solver indicates that the statement of Theorem~\ref{thm:main} does not hold for $k = 2$, and that there exists a proportional and strategyproof rule, at least for small parameter values. However, we can recover an impossibility by strengthening strategyproofness to superset-strategyproofness, i.e., by allowing manipulators to report arbitrary ballots (rather than only subsets of the truthful ballot).

\begin{theorem}
	\label{thm:k-2}
	Let $k = 2$, $m \ge 4$, and let $n$ be even. Then there is no approval-based committee rule that satisfies weak efficiency, JR on party lists and superset-strategyproofness.
\end{theorem}

The proof of this result was also obtained via the computer-aided method. However, this proof is long and involves many case distinctions, so we omit the details. The proof begins with the starting profile $P = (ab,ab,cd,cd)$. By JR on party lists, we have $f(P) \in \{ac,ad,bc,bd\}$. By relabeling alternatives, we may assume that $f(P) = ac$. The proof then applies strategyproofness to deduce the values of $f$ at other profiles, and arrives at a contradiction.

Theorem~\ref{thm:k-2} requires an even number of voters. This is necessary, since for $k=2$ and odd numbers of voters, AV satisfies both axioms.

\begin{proposition}
	\label{prop:av-odd}
	For $k = 2$, any $m \ge 3$, and $n$ odd, AV satisfies proportionality (it even satisfies JR) and is cardinality-strategyproof.
\end{proposition}
\begin{proof}
	AV is cardinality-strategyproof. \citet[Thm.~3]{ejr} showed that for $k = 2$ and odd $n$, $\textup{AV}$ satisfies JR. For completeness, we repeat their argument here. Let $P$ be a profile. Suppose there is some group $N' \subseteq N$ with $|N'| \ge \frac nk$ with $c\in P(i)$ for all $i\in N'$. Note that $|N'| \ge \frac nk$ implies $|N'| > \frac n2$, so that $c$ has approval score $> \frac n2$. Then the highest approval score is also $> \frac n2$, and so there is some $d\in \textup{AV}(P)$ with approval score $> \frac n2$. Thus, a strict majority of voters approve $d$. Since strict majorities intersect, there must be a voter $i\in N'$ who approves $d$. Thus $d \in \textup{AV}(P) \cap \bigcup_{i\in N'} P(i)$, whence the latter set is non-empty, and JR is satisfied.
\end{proof}

\section{Related Work}

The closest work to ours is a short article by \citet{duddy}, who also proves an impossibility about approval-based committee rules involving a proportionality axiom. Duddy's result is about \emph{probabilistic} committee rules, which return probability distributions over the set of committees. Because any deterministic committee rule induces a probabilistic one (which puts probability 1 on the deterministic output), Duddy's probabilistic result also has implications for deterministic rules, which we can state as follows.

\begin{theorem}[Duddy \citep{duddy}]
	For $m=3$ and $k=2$, no approval-based committee rule $f$ satisfies the following three axioms.
	\begin{enumerate}
		\item (Representative.) There exists a profile $P$ in which $n$ voters approve $\{x\}$ and $n+1$ voters approve $\{y,z\}$, but $f(P) \neq \{y,z\}$, for some $n \in \mathbb N$ and all distinct $x,y,z\in C$.
		\item (Pareto-consistent.) If in profile $P$, the set of voters who approve of $x$ is a strict subset of the set of voters who approve of $y$, then $f(P) \neq \{x,z\}$, for all distinct $x,y,z\in C$.
		\item (Strategyproof.) Suppose profiles $P$ and $P'$ are identical, except that voter $i$ approves $\{x,y\}$ in $P$ but $\{x\}$ in $P'$. If $f(P) \neq \{x,y\}$, then also $f(P') \neq \{x,y\}$.
	\end{enumerate}
\end{theorem}
How does Duddy's theorem relate to ours? Duddy's strategyproofness is weaker than but very similar to our strategyproofness. Our result does not require an efficiency axiom. Duddy's representative axiom is noticeably different from the proportionality axioms that we have discussed. Logically it is incomparable to our proportionality axiom; in spirit it may be slightly stronger. Note that not even the strongest of the proportionality axioms that we have discussed (i.e., EJR) imply Duddy's representativeness. It is also worth noting that Duddy's result works for smaller values of $m$ and $k$ than our result, suggesting that Duddy's axioms are stronger overall. 

In computational social choice, there has been much recent interest in axiomatic questions in committee rules. Working in the context of strict orders, \citet{elk-fal-sko-sli:c:multiwinner-rules} introduced several axioms and studied which committee rules satisfy them. 
\citet{corr_SkowronFS16} axiomatically characterise the class of \emph{committee scoring rules}, and \citet{fal-sko-sli-tal:c:classification} study the finer structure of this class.
For the approval-based setting, \citet{lac-sko:t:abc-approval-multiwinner} characterise \emph{committee counting rules}, and give characterisations of PAV and of Chamberlin--Courant. They also have a result suggesting that AV is the only consistent committee rule which is strategyproof.

From a computational complexity perspective, there have been several papers studying the complexity of manipulative attacks on multiwinner elections \citep{meir2008complexity,obraztsova2013manipulation,faliszewski2017bribery,aziz2015computationalaspects,baumeister2015winner}. 
Other work has studied the complexity of evaluating various committee rules. Notably, it is NP-complete to find a winning committee for PAV \citep{aziz2015computationalaspects,skowron2015achieving}.

\section{Conclusions and Future Work}

We have proved an impossibility about approval-based committee rules. The versions of the proportionality and strategyproofness axioms we used are very weak. It seems unlikely that, by weakening the axioms used, one can find a committee rule that exhibits satisfying versions of these requirements.  A technical question which remains open is whether our impossibility holds for \emph{all} numbers $n$ of voters, no matter whether it is a multiple of $k$ (see Section~\ref{sec:droop}). It would also be interesting to study irresolute or probabilistic rules.

To circumvent the classic impossibilities of Arrow and Gibbard--Satterthwaite, it has proved very successful to study restricted domains such as single-peaked preferences, which can often give rise to strategyproof voting rules \citep{Moul88a,ELP17a}.
\citet{elkind2015structure} propose analogues of single-peaked and single-crossing preferences for the case of approval ballots and dichotomous preferences. For example, a profile of approval ballots satisfies the Candidate Interval (CI) condition if there exists an underlying linear ordering of the candidates such that each voter approves an \emph{interval} of candidates \citep[see also][]{FHHR11a}.
Restricting the domain to CI profiles in our SAT encoding suggests that an impossibility of the type we have studied cannot be proven for this domain -- at least for small values of $n$, $m$, and $k$.
Finding a proportional committee rule that is not manipulable on the CI domain would be an exciting avenue for future work.

It would be interesting to obtain impossibilities  using other axioms.
Recently, \citet{sanchez2017monotonicity} found some incompatibilities between proportionality and \emph{monotonicity}. Their version of proportionality (`perfect representation'), however, is very strong and possibly undesirable. It would be interesting to see whether such results hold for weaker versions of their axioms.

\paragraph{Acknowledgements.}
I thank Markus Brill for suggesting subset manipulations and other ideas, Martin Lackner and Piotr Skowron for useful discussions, and reviewers from AAMAS and COMSOC for helpful comments. This work was supported by EPSRC and ERC grant 639945 (ACCORD).


\bibliographystyle{ACM-Reference-Format}  


\end{document}